\documentclass{amsart}

\usepackage[T1]{fontenc}
\usepackage[margin=1in]{geometry}
\usepackage{amsmath,amssymb,amsthm}
\usepackage{mathtools}\mathtoolsset{centercolon}
\usepackage{tikz}
\makeatletter
\protected\def\tikz@nonactivecolon{\ifmmode\mathrel{\mathop\ordinarycolon}\else:\fi} 
\makeatother
\usetikzlibrary{decorations.pathreplacing}
\usepackage{verbatim}
\usepackage[normalem]{ulem} 
\usepackage[colorlinks]{hyperref}
\usepackage[all]{hypcap}
\usepackage{thmtools}
\usepackage{thm-restate}
\usepackage[foot]{amsaddr}
\usepackage{color}

\newcommand{\eq}[1]{\hyperref[eq:#1]{(\ref*{eq:#1})}}
\renewcommand{\sec}[1]{\hyperref[sec:#1]{Section~\ref*{sec:#1}}}
\newcommand{\app}[1]{\hyperref[app:#1]{Appendix~\ref*{app:#1}}}
\newcommand{\thm}[1]{\hyperref[thm:#1]{Theorem~\ref*{thm:#1}}}
\newcommand{\lem}[1]{\hyperref[lem:#1]{Lemma~\ref*{lem:#1}}}
\newcommand{\defn}[1]{\hyperref[defn:#1]{Definition~\ref*{defn:#1}}}
\newcommand{\ex}[1]{\hyperref[ex:#1]{Example~\ref*{ex:#1}}}
\newcommand{\fct}[1]{\hyperref[fct:#1]{Fact~\ref*{fct:#1}}}
\newcommand{\fig}[1]{\hyperref[fig:#1]{Figure~\ref*{fig:#1}}}
\newcommand{\figs}[1]{\hyperref[fig:#1]{Figures~\ref*{fig:#1}}}

\theoremstyle{plain}

\newtheorem{lemma}{Lemma}

\newcommand{\sket}[1]{|{#1}\rangle}
\newcommand{\sbraket}[2]{\langle{#1}|{#2}\rangle}

\newcommand{\CC}{\mathbb{C}}
\newcommand{\II}{\mathbb{I}}

\newcommand{\QQ}{\mathbb{Q}}

\newcommand{\ZZ}{\mathbb{Z}}
\newcommand{\posint}{\ZZ^+}
\renewcommand{\sc}{\mathrm{sc}}
\newcommand{\bc}{\mathrm{bc}}
\newcommand{\ph}{\mathrm{ph}}

\DeclareMathOperator{\spn}{span}

\title{Momentum switches}
\author{Andrew M. Childs$^{1,2}$} \email{amchilds@uwaterloo.ca}
\author{David Gosset$^{1,2}$} \email{dngosset@gmail.com}
\author{Daniel Nagaj$^{3,4}$} \email{daniel.nagaj@univie.ac.at}
\author{Mouktik Raha$^{2,5}$} \email{mouktikraha@gmail.com}
\author{Zak Webb$^{2,6}$} \email{zakwwebb@gmail.com}

\address{$^1$ Department of Combinatorics \& Optimization, University of Waterloo}
\address{$^2$ Institute for Quantum Computing, University of Waterloo}
\address{$^3$ Faculty of Physics, University of Vienna}
\address{$^4$ Institute of Physics, Slovak Academy of Sciences}
\address{$^5$ Department of Physics and Meteorology, Indian Institute of Technology Kharagpur}
\address{$^6$ Department of Physics \& Astronomy, University of Waterloo}

\date{}

\begin{document}

\maketitle 

\begin{abstract}
Certain continuous-time quantum walks can be viewed as scattering processes. These  processes can perform quantum computations, but it is challenging to design graphs with desired scattering behavior.  In this paper, we study and construct momentum switches, graphs that route particles depending on their momenta.  We also give an example where there is no exact momentum switch, although we construct an arbitrarily good approximation.
\end{abstract}

\section{Introduction}
\label{sec:intro}

Quantum walk is a powerful tool for quantum computation.  In particular, the concept of scattering on graphs has been used to develop algorithms \cite{FG98,FGG07} and to establish universality of models of computation based on quantum walk \cite{Chi09,CGW13}.

In the scattering framework, we consider an infinite graph obtained by attaching semi-infinite paths to some of the vertices of a finite graph $\hat{G}$, as shown in \fig{basic_scattering}. A particle is initialized in a state that moves toward $\hat{G}$ on one of the semi-infinite paths. After some time the particle has scattered; it moves away from $\hat{G}$ and, in general, has some outgoing amplitude on each of the semi-infinite paths.  By choosing the graph carefully, such a scattering process can be designed to perform a quantum computation.

A discrete version of scattering theory can be used to compute the amplitude scattered into each path. The theory of scattering on graphs was introduced by Farhi and Gutmann in the setting with two semi-infinite paths \cite{FG98}; Childs presented an application with an arbitrary number of semi-infinite paths \cite{Chi09}.  Other work has described further basic properties of scattering on graphs \cite{VB09}, classified the scattering properties of some small graphs using a computer search \cite{BUF11}, established a discrete analog of Levinson's Theorem \cite{CS11,CG12}, and proved completeness of the scattering and bound states \cite{CG12}.

While it is straightforward to compute the scattering behavior of a given graph, it is considerably more difficult to design a graph that implements some desired scattering behavior.  Our goal in this paper is to develop tools for constructing scattering gadgets.  We hope that these ideas will ultimately prove useful in the design of scattering algorithms. 

We focus on a scattering gadget called a \emph{momentum switch}.  A momentum switch has three terminals (i.e., has the form of \fig{basic_scattering} with $N=3$) and has special scattering properties for (at least) two momenta $k$ and $p$.  A particle with momentum $k$ transmits perfectly between paths $1$ and $2$, whereas a particle with momentum $p$ transmits perfectly between paths $1$ and $3$. Thus a momentum switch routes a particle in a direction that depends on its momentum, as shown in \fig{qual_momentum_switch}.


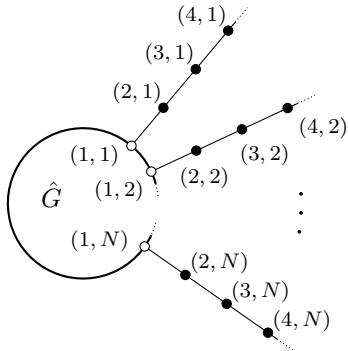
\begin{figure}
\begin{tikzpicture}[
  label distance=-5.5pt,
  thin,
  vertex/.style={circle,draw=black,fill=black,inner sep=1.25pt,
    minimum size =0mm},
  attach/.style={circle,draw=black,fill=white,inner sep=1.25pt,
    minimum size =0mm},
  dots/.style={circle,fill=black,inner sep=.5pt,
    minimum size= 0pt},
  every text node part/.style={font=\footnotesize}]

  \node at (.15, .63) [rectangle,fill=white] {$(1,1)$};
  \node at (.48, .17) [rectangle,fill=white,inner sep=0pt] {$(1,2)$};
  \node at (.22,-.5) [rectangle,fill=white] {$(1,N)$};

  \draw[thick] (15:1) arc (15:335:1);
  \draw[densely dotted] (15:1)  arc (15:5:1);
  \draw[densely dotted] (-15:1)  arc (-15:-25:1);

  \node at (-.42,.12) [rectangle] {\normalsize${\hat{G}}$};

  \foreach \x  in {50, 25,-35}{
    \draw (\x:1cm) -- (\x:3.15cm);
    \draw (\x:2.9cm) -- (\x:3.4cm) [densely dotted];
  }
  
  \node at (50:1)[attach]{};
  \node at (50:1.66)[vertex,label=150:{$(2,1)$}]{};
  \node at (50:2.33)[vertex,label=150:{$(3,1)$}]{};
  \node at (50:3)[vertex,label=150:{$(4,1)$}]{};

  \node at (25:1)[attach]{};
  \node at (25:1.66)[vertex]{};
  \node at (25:2.33)[vertex]{};
  \node at (25:3)[vertex]{};
  
  \node at (12:1.65) {$(2,2)$};
  \node at (15:2.5) {$(3,2)$};
  \node at (17.5:3.35) {$(4,2)$};

  \node at (-35:1)[attach]{};
  \node at (-35:1.66)[vertex,label=60:{$(2,N)$}]{};
  \node at (-35:2.33)[vertex,label=60:{$(3,N)$}]{};
  \node at (-35:3)[vertex,label=60:{$(4,N)$}]{};

  \node at (-2.5:2.9)[dots] {};
  \node at (2.5:2.9)[dots] {};
  \node at (-7.5:2.9)[dots] {};
\end{tikzpicture}
\caption{A finite graph $\hat{G}$ with $N$ semi-infinite attached. The open circles are \emph{terminals}, vertices of $\hat{G}$ to which semi-infinite paths are attached. The internal vertices of $\hat{G}$ are not shown.}
\label{fig:basic_scattering}
\end{figure}

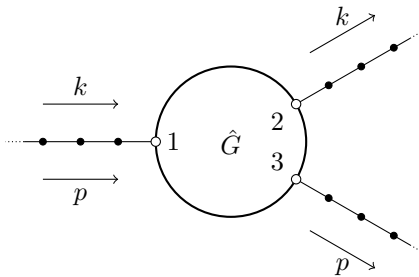
\begin{figure}
\begin{tikzpicture}
  [ thin,
    inner/.style={circle,draw=black!100,fill=black!100,inner sep = .9 pt},
  attach/.style={circle,draw=black,fill=white,inner sep=1.25pt,
    minimum size =0mm}]

  \draw (0,0) [thick] circle (1cm);
\node at (0,0) {$\hat{G}$};
\node at (180:1.1) [label=right:{$1$}] {};
\node at (30:1.1) [label=200:{$2$}]{};
\node at (330:1.1) [label=160:{$3$}]{};

\foreach \t in {180,30,330}{
  \draw (\t:1) -- (\t:2.75);
  \draw[densely dotted] (\t:2.75) -- (\t:3);
  \node[attach] at (\t:1) {}; 
  \foreach \x in {1.5,2,2.5}{
    \node[inner] at (\t:\x) {};
}}

\begin{scope}[xshift = -2.5cm,yshift = .5cm]
  \draw[->] (0,0) -- (1,0) node [midway,above] {$k$};
\end{scope}

\begin{scope}[rotate=30,xshift = 1.5cm, yshift = .5cm]
  \draw[->] (0,0) -- (1,0) node [midway,above] {$k$};
\end{scope}

\begin{scope}[xshift = -2.5cm,yshift=-.5cm]
  \draw[->] (0,0) -- (1,0) node [midway,below] {$p$};
\end{scope}

\begin{scope}[rotate=-30,xshift = 1.5cm,yshift=-.5cm]
  \draw[->] (0,0) -- (1,0) node [midway,below] {$p$};
\end{scope}

\end{tikzpicture}
\caption{A momentum switch.  A particle moving toward vertex 1 with momentum $k$ transmits perfectly to the upper path (through vertex 2), while a particle with momentum $p$ transmits to the lower path (through vertex 3).}
\label{fig:qual_momentum_switch}
\end{figure}

A switch between momenta $-\frac{\pi}{2}$ and $-\frac{\pi}{4}$ was used as a tool in the multi-particle quantum walk universality construction \cite{CGW13}. In this paper, we construct switches between other pairs of momenta. To achieve this, we consider a closely related type of graph called a reflection/transmission (R/T) gadget. An R/T gadget is a graph with two terminals (as in \fig{basic_scattering} with $N=2$) such that some momenta transmit perfectly between the two paths, whereas other momenta perfectly reflect. The momentum switches we construct in this paper are built by combining R/T gadgets in a prescribed way.

We also show that it is not always possible to construct a momentum switch.  In particular, we prove that there is no switch between momenta $-\frac{\pi}{4}$ and $-\frac{3\pi}{4}$.  (These two particular momenta are relevant not only because they provide a concrete limitation on the construction of momentum switches, but because they both support the universal gates constructed in \cite{Chi09}, so a momentum switch between them would simplify a multi-particle universality construction along the lines of \cite{CGW13}.)  Nevertheless, we exhibit graphs that approximate a momentum switch at these two momenta to arbitrarily high precision. 

The remainder of this paper is organized as follows. In \sec{scat_theory} we review scattering theory on graphs. Then in \sec{mswitch} we define momentum switches and R/T gadgets. In \sec{reflection} we give some explicit constructions of R/T gadgets, and in \sec{switch} we describe how to construct a momentum switch starting from a specific type of R/T gadget.  Using this construction, we obtain a large class of momentum switches. In \sec{impossibility} we prove that there is no perfect momentum switch between $-\frac{\pi}{4}$ and $-\frac{3\pi}{4}$, and in \sec{approx_switch} we describe approximate momentum switches between these momenta. We conclude in \sec{conc} with a discussion of the results and some directions for future work.

\section{Continuous-time quantum walk and scattering theory}
\label{sec:scat_theory}

The continuous-time quantum walk on an unweighted graph $G$ with vertex set $V(G)$ lives in the Hilbert space $\spn\{|v\rangle\colon v\in V(G)\}$ and is generated by the time-independent Hamiltonian equal to the adjacency matrix of the graph.

First consider the case where $G$ is an infinite path, so the Hamiltonian is
\[
H=\sum_{x\in\ZZ} \bigl(|x\rangle\langle x+1|+|x+1\rangle\langle x|\bigr).
\]
As with a free particle in one dimension (in the continuum), this Hamiltonian does not have any normalized eigenvectors. However, if we allow unnormalized states, then we can solve the eigenvalue equation and obtain eigenvectors $|\tilde{k}\rangle$ for each $k\in [-\pi,\pi)$, defined by $\langle x |\tilde{k}\rangle = e^{-ikx}$. These states satisfy $\langle x|H|\tilde{k}\rangle= E(k)\langle x|\tilde{k}\rangle$, where
\[
E(k)=2\cos(k).
\]
We call these momentum states; the number $k\in [-\pi,\pi)$ is the corresponding momentum.  A wave packet consisting of momenta near $k$ moves with speed $|\frac{dE}{dk}|=|2\sin(k)|$.

Now consider the more general setup shown in \fig{basic_scattering}. In this setting one can prepare a particle with momentum near $k$ on the semi-infinite path labeled $j\in [N] := \{1,\ldots,N\}$; if $k \in (-\pi,0)$, the particle moves toward the finite graph $\hat{G}$. After some time the particle will be in a superposition of states that move away from $\hat G$ on the semi-infinite paths. Such processes are described by the incoming scattering eigenstates 
\[
\{\sket{\sc_j (k)} \colon k\in(-\pi,0), j\in [N]\}
\]
which satisfy $H\sket{\sc_j (k)}=2\cos(k)\sket{\sc_j (k)}$. Labeling vertices on the paths by $(x,j')$ (where $j' \in [N]$ labels the path and $x \in \posint := \{1,2,\ldots\}$ labels the location on the path), these states have the form
\begin{equation}
  \langle x,j' \sket{\sc_j (k)} = \delta_{j'\!,j} e^{-ikx} + S_{j'\!,j}(k) e^{ikx},
\label{eq:scat_states}
\end{equation}
where $S(k)$ is a unitary matrix called the S-matrix.  Since $H$ is real, $S(k)$ is also symmetric, as can be seen from equation (2.8) of \cite{CG12}.  The matrix element $S_{j'\!,j}(k)$ can be interpreted as the amplitude for a wave packet with momentum $k$ to scatter from the $j$th semi-infinite path to the $j^{\prime}$th. 

In this paper we are only concerned with the scattering eigenstates of $H$, since they are sufficient to understand the scattering dynamics.  However, $H$ can also have other eigenvectors; a detailed description can be found in reference \cite{CG12}.

Equation \eq{scat_states} only specifies the form of the scattering state $\sket{\sc_j (k)}$ on the semi-infinite paths. To determine the state, one must compute the S-matrix at momentum $k$ as well as the amplitudes $\langle v\sket{\sc_j (k)}$ for vertices $v\in\hat{G}$. These quantities can be computed from $\hat{G}$ using the eigenvalue equation $H\sket{\sc_j (k)}=2\cos(k)\sket{\sc_j (k)}$.  This condition gives $|V(\hat G)|$ linear equations that specify the $N$ entries in the $j$th row of the S-matrix and the $|V(\hat G)|-N$ amplitudes at the internal vertices (those to which semi-infinite paths are not attached).  Reference \cite{CG12} gives more details on computing the S-matrix.

\section{Momentum switches and reflection/transmission gadgets}
\label{sec:mswitch}

As discussed in \sec{intro}, a momentum switch is a special type of gadget that can be used to route a particle depending on its momentum. This property is naturally described in terms of the S-matrix. 

For example, \fig{mswitch} shows the momentum switch used in reference \cite{CGW13}. The S-matrix of this graph has a special form at momenta $-\frac{\pi}{4}$ and $-\frac{\pi}{2}$:
\begin{equation}
  S_{\mathrm{switch}}(-\tfrac{\pi}{4}) = \begin{pmatrix} 0 & 0 & e^{-i\pi/4}\\
    0 & -1 & 0\\
    e^{-i\pi/4} & 0 & 0\end{pmatrix}\qquad
  S_{\mathrm{switch}}(-\tfrac{\pi}{2}) = \begin{pmatrix}0 & -1 &0\\
    -1 & 0 & 0\\
    0 & 0 & 1\end{pmatrix}.
\label{eq:switch_matrices}
\end{equation} 
This equation says that a particle with momentum $-\frac{\pi}{4}$ traveling towards the graph along path $1$ transmits perfectly to path $3$ (i.e., the amplitude for this process has unit magnitude), whereas a particle with momentum $-\frac{\pi}{2}$ traveling along path $1$ transmits perfectly to path $2$. In other words, this graph is a momentum switch between momenta $-\frac{\pi}{4}$ and $-\frac{\pi}{2}$. 

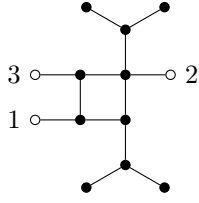
\begin{figure}
\centering
\begin{tikzpicture}
  [ scale = 0.6,
    thin,
    inner/.style={circle,draw=black!100,fill=black!100,inner sep=1.25pt},
    attach/.style={circle,draw=black!100,fill=black!0,thin,inner sep=1.25pt},
    sin/.style={line width=.7pt},
    doub/.style={line width=2.1pt},
    trip/.style={draw=white,line width=.7pt}]
    \node at (0,0){};
\begin{scope}[yshift=1.8cm]
  \node (2)  at (0,0) [attach,label=left:1] {};
  \node (1)  at (0,1) [attach,label=left:3] {};
  \node (3)  at (3,1) [attach,label=right:2] {};
  \node (4)  at (1,0) [inner]  {};
  \node (5)  at (2,0) [inner]  {};
  \node (6)  at (1,1) [inner]  {};
  \node (7)  at (2,1) [inner]  {};
  \node (8)  at (2,2) [inner]  {};
  \node (9)  at (2.866,2.5)  [inner] {};
  \node (10) at (1.134,2.5)  [inner] {};
  \node (11) at (2,-1)       [inner] {};
  \node (12) at (2.866,-1.5) [inner] {};
  \node (13) at (1.134,-1.5) [inner] {};

  \draw (2) to (4);
  \draw (4) to (5);
  \draw (3) to (7);
  \draw (1) to (6);
  \draw (6) to (4);
  \draw (6) to (7);
  \draw (7) to (5);
  \draw (7) to (8);
  \draw (8) to (9);
  \draw (8) to (10);
  \draw (11) to (5);
  \draw (11) to (12);
  \draw (11) to (13);

  \end{scope}
\end{tikzpicture}
\caption{The momentum switch from \cite{CGW13}.}
\label{fig:mswitch}
\end{figure}

Generalizing this example, a finite graph $\hat{G}$ with three terminals (labeled $1,2,3$) is a momentum switch between two (disjoint) sets of momenta $\mathcal{D},\mathcal{D}' \subset  (-\pi,0)$ if its S-matrix has perfect transmission from terminal $1$ to terminal $2$ at each momentum $k\in \mathcal{D}$ and perfect transmission from terminal $1$ to terminal $3$ at each momentum $p\in \mathcal{D}'$ (i.e., $|S_{1,2}(k)| = |S_{1,3}(p)|=1$ for $k\in \mathcal{D}$ and $p\in \mathcal{D}'$).  See \fig{qual_momentum_switch} for an illustration.

Momentum switches are closely related to another class of graphs that we call reflection/transmission (R/T) gadgets. An R/T gadget is a finite graph with two terminals. In addition, there exist two sets of momenta $\mathcal{R},\mathcal{T}\subset (-\pi,0)$ such that the gadget perfectly reflects (from both terminals) at all $k\in \mathcal{R}$ and perfectly transmits (between the two terminals) at all $p\in \mathcal{T}$ (i.e., $|S_{1,1}(k)| = |S_{1,2}(p)| = 1$ for $k\in \mathcal{R}$ and $p\in \mathcal{T}$).

There is a simple connection between R/T gadgets and momentum switches.  While a momentum switch has three terminals and an R/T gadget only has two, we now show that if we downgrade one of the terminals of a momentum switch to an internal vertex, the switch becomes an R/T gadget (between the momenta it separated). 

Let $\hat{G}$ be a momentum switch and fix $k\in \mathcal{D}$ and $p\in \mathcal{D}'$. The S-matrix takes the form
\[
  S_{\mathrm{switch}}(k) = \begin{pmatrix} 0 &T & 0\\
  T & 0 & 0\\
  0 & 0 & R\end{pmatrix} \qquad
  S_{\mathrm{switch}}(p) = \begin{pmatrix} 0 & 0 & T'\\
  	0 & R' & 0\\
	T' & 0 & 0\end{pmatrix},
\]
i.e., the switch connects paths $1$ and $2$ at momentum $k$ and paths $1$ and $3$ at momentum $p$.  Using equation \eq{scat_states}, we see that the states $\sket{\sc_1(k)}$, $\sket{\sc_2(k)}$, and $\sket{\sc_{2}(p)}$ have no amplitude on path 3.

Let $G'$ be the graph obtained from $G$ by removing the semi-infinite path connected to terminal 3, i.e., now we only attach semi-infinite paths to terminals $1$ and $2$. Since the states $\sket{\sc_1(k)}$, $\sket{\sc_2(k)}$, and $\sket{\sc_{2}(p)}$ have no amplitude on the removed vertices, they remain scattering eigenstates. The S-matrix of $G'$ has the form
\[
  S_{\text{R/T}}(k) = \begin{pmatrix} 0 & T\\ T & 0\end{pmatrix} \qquad S_{\text{R/T}}(p) = \begin{pmatrix} R'' & T''\\ 0 & R'\end{pmatrix},
\]
where $R''$ and $T''$ need to be determined. Unitarity implies $T''=0$ and thus
\[
  S_{\text{R/T}}(p) = \begin{pmatrix} R'' & 0\\
  	0 & R'\end{pmatrix}.
\]
Hence $G'$ is an R/T gadget with $\mathcal{D} \subseteq \mathcal{T}$ and $\mathcal{D}' \subseteq \mathcal{R}$. The same construction can be used to obtain an R/T gadget with $\mathcal{D}' \subseteq \mathcal{T}$ and $\mathcal{D} \subseteq \mathcal{R}$ (by downgrading terminal 2 instead of terminal 3).

\section{Constructing R/T gadgets}
\label{sec:reflection}

In this Section we show how to design R/T gadgets between certain sets of momenta.  We construct gadgets of the special form shown in \fig{reversal_orig}, which we call type 1 R/T gadgets.  We focus on such gadgets because their scattering properties are closely related to the eigenvectors of the subgraph $G_0$.

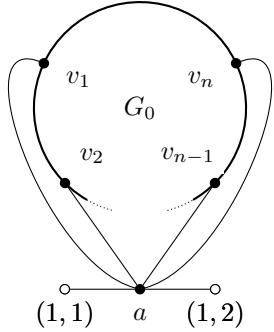
\begin{figure}
\centering
\begin{tikzpicture}[
  label distance=-5.5pt,
  thin,
  vertex/.style={circle,draw=black,fill=black,inner sep=1.25pt,
    minimum size =0mm},
  attach/.style={circle,draw=black,fill=white,inner sep=1.25pt,
    minimum size =0mm},
  dots/.style={circle,fill=black,inner sep=.5pt,
    minimum size= 0pt}]

  \draw (-1,-2.41) -- (1,-2.41);
  
  \foreach \x /\n in {-1/ 1, 1/2}{
    \node at (\x,-2.41) [attach] {};
    \node at (\x,-2.75) {$(1,\n)$};
  }
    
  \node (a) at (0,-2.41) [vertex] {};
  \node at (0,-2.75) {$a$};
  
  \foreach \x /\n in {-1/ 1, 1/2}{
    \node at (\x,-2.41) [attach] {};
    \node at (\x,-2.75) {$(1,\n)$};
  }
  
  \draw[thick] (240:1.41) arc (240:-60:1.41);
  \draw[densely dotted] (240:1.41)  arc (240:255:1.41);
  \draw[densely dotted] (285:1.41)  arc (285:300:1.41);
  
  \foreach \i / \n /\t in {25 / n/ 15, 155 / 1/ 165, 225/ 2 / 120, 315 / {n-1}/60}{
    \node at (\i : .9) [rectangle,fill=white] {$v_\n$};
    \node at (\i : 1.41) [vertex] {};
  }
  \node at (0,0) [rectangle,fill=white] {$G_0$};

  \foreach \i /\t in {25/15, 155/165}{
    \draw (\i:1.41) to[out=\i,in=\t] (a);
  }
  
  \foreach \i \in in {225,315}{
    \draw (\i:1.41) to (a);
  }

\end{tikzpicture}
\caption{A type 1 R/T gadget. In the special case where there is only one edge between $G_0$ and vertex $a$ (i.e., when $n=1$) this is also a type 2 R/T gadget.}
\label{fig:reversal_orig}
\end{figure}

We refer to the graph shown in \fig{reversal_orig} as $\hat{G}$, and we write $G$ for the full graph obtained by attaching two semi-infinite paths to terminals $(1,1)$ and $(1,2)$.  As shown in the Figure, the graph $\hat{G}$ contains a finite subgraph $G_0$ that is connected to the vertex labeled $a$ with edges to vertices in $S=\{v_1,\ldots,v_n\}$. We write $g_0$ for the induced subgraph on $V(G_0)\setminus S$. If a type 1 gadget has $|S|=1$ (as shown in \fig{reversalRT}) then we call it a type 2 gadget.

Looking at the eigenvalue equation for the scattering state $\sket{\sc_{1} (k)}$ at vertices $(1,1)$ and $(1,2)$, we see that the amplitude at vertex $a$ satisfies
\[
  \langle{a}\sket{\sc_{1} (k)} = 1 + R(k) = T(k).
\] 
Thus perfect reflection at momentum $k$ occurs if and only if $R(k)=-1$ and $\langle{a}\sket{\sc_{1} (k)}=0$, while perfect transmission occurs if and only if $T(k)=1$ and $\langle{a}\sket{\sc_{1} (k)}=1$. Using this fact, we now derive conditions on the graph $G_0$ that determine when perfect transmission and reflection occur.

For type 1 gadgets, we give a necessary and sufficient condition for perfect reflection:

\begin{lemma}\label{lem:reflect_reqs}  Let $\hat{G}$ be a type 1 R/T gadget. A momentum $k\in (-\pi,0)$ is in the reflection set $\mathcal{R}$ if and only if $G_0$ has an eigenvector $\sket{\chi_k}$ with eigenvalue $2\cos(k)$ satisfying
\begin{equation}
  \sum_{i=1}^{n} \sbraket{v_i}{\chi_k} \neq 0. \label{eq:sum_condition}
\end{equation}
\end{lemma}

\begin{proof}
First suppose that $\hat{G}$ has perfect reflection at momentum $k$, i.e., $R(k)=-1$ and $\langle{a}\sket{\sc_{1} (k)}=0$. Since $\langle{(1,1)}\sket{\sc_1(k)} = e^{-ik} - e^{ik}\neq 0$ and $\langle{(1,2)}\sket{\sc_1(k)}=0$, to satisfy the eigenvalue equation at vertex $a$, we have
\[
  \sum_{j=1}^{n} \langle{v_j}\sket{\sc_1(k)} = e^{ik} - e^{-ik} \neq 0.
\]
Further, since $G_0$ only connects to vertex $a$ and the amplitude at this vertex is zero, the restriction of $\sket{\sc_1(k)}$ to $G_0$ must be an eigenvector of $G_0$ with eigenvalue $2\cos(k)$. Hence the condition is necessary for perfect reflection. 
 
Next suppose that $G_0$ has an eigenvector $\sket{\chi_k}$ with eigenvalue $2\cos(k)$ satisfying \eq{sum_condition}, with the sum equal to some nonzero constant $c$. Define a state $\sket{\psi_k}$ on the Hilbert space of the full graph $G$ with amplitudes
\[
  \langle{w} \sket{\psi_k} = \frac{e^{ik} - e^{-ik}}{c} \langle{w} \sket{\chi_k}
\]
for all $w \in V(G_0)$, $\langle a|\psi_k\rangle=0$, and 
\[
 \langle{(x,j)} \sket{\psi_k}=\begin{cases} e^{-ikx}-e^{ikx} & j=1\\
0 & j=2
\end{cases}
\]
for $x \in \posint$. One can verify that $|\psi_k\rangle$ is an eigenvector of $G$ with eigenvalue $2\cos (k)$, and takes the form of a scattering eigenstate with perfect reflection. This shows that $R(k)=-1$ and $T(k)=0$ as claimed. (If the state $\sket{\chi_k}$ satisfying the above conditions is unique then $\sket{\sc_1(k)}=|\psi_k\rangle$; otherwise they are equal for an appropriate choice of $\sket{\chi_k}$.)\end{proof}

The following Lemma gives a sufficient condition for perfect transmission (which is also necessary for type 2 gadgets).  Recall that $g_0$ is the induced subgraph on $V(G_0)\setminus S$.

\begin{lemma}\label{lem:transmit_reqs}
Let $\hat{G}$ be a type 1 R/T gadget and let $k\in (-\pi,0)$. Suppose $\sket{\xi_k}$ is an eigenvector of $g_0$ with eigenvalue $2\cos{k}$ and with the additional property that, for all $i \in [n]$,
\begin{equation}
\label{eq:trans_cond}
  \sum_{\substack{u\in V(g_0): \\ (u,v_i)\in E(G_0)}} \langle{u}\sket{\xi_k} = c \neq 0 
\end{equation}
for some constant $c$ that does not depend on $i$. Then $k$ is in the transmission set $\mathcal{T}$. If $\hat{G}$ is a type 2 R/T gadget, then this condition is also necessary.
\end{lemma}

\begin{proof}
If $g_0$ has a suitable eigenvector $\sket{\xi_k}$ satisfying \eq{trans_cond}, define a state $\sket{\psi_k}$ on the full graph $G$, with amplitudes $\langle a\sket{\psi_k}=1$, 
\[
  \langle w \sket{\psi_k} = \begin{cases} -\frac{1}{c} \langle{w}\sket{\xi_k} & w\in V(g_0)\\
  	0 & w\in S
\end{cases}
\]
in the graph $G_0$, and 
\[
 \langle{(x,j)} \sket{\psi_k}=\begin{cases} e^{-ikx} & j=1\\
 e^{ikx} & j=2
\end{cases}
\]
for $x \in \posint$.  Since $\sket{\xi_k}$ satisfies equation \eq{trans_cond}, the eigenvalue equation is satisfied for each vertex in $S$, so the state $\sket{\psi}$ is an eigenvector of $G$ with eigenvalue $2\cos(k)$ and perfect transmission, which shows that $T(k)=1$. 

Now suppose $\hat{G}$ is a type 2 R/T gadget, with $S = \{v\}$.  Perfect transmission along with the eigenvalue equation at vertex $a$ implies
\[
\langle{v}\sket{\sc_1(k)} = 0.
\]
Now applying the eigenvalue equation at $v$, we get
\[
  \sum_{u:(u,v)\in E(G_0)} \langle u \sket{\sc_1(k)} = -1.
\]
Hence the restriction of $\sket{\sc_1(k)}$ to $V(g_0)$ is an eigenvector of the induced subgraph, with the additional property that the sum of the amplitudes of vertices connected to $v$ is nonzero (and trivially the same for each $v\in S$).\end{proof}

\subsection{Reflection/transmission set reversal}
\label{sec:reversal}

We now show how to switch the reflection and transmission sets for a type 2 gadget. In particular, for any such gadget with transmission set $\mathcal{T}$ and reflection set $\mathcal{R}$, we construct another (type 1) gadget with transmission set $\mathcal{T}'$ and reflection set $\mathcal{R}$ such that $\mathcal{R} \subseteq \mathcal{T}'$ and $ \mathcal{T} \subseteq \mathcal{R}'$.

This new R/T gadget $\hat{G}^{\leftrightarrow}$ is depicted in \fig{reversal_changed}. It is obtained by taking two copies of the subgraph $g_0$ from \fig{reversalRT}, connecting both to the vertex $v$, and then connecting one copy of $g_0$ to the infinite path.  Note that for each vertex $w_j$ in \fig{reversalRT} there are two corresponding vertices $w^{(1)}_j,w^{(2)}_j$ in \fig{reversal_changed}.

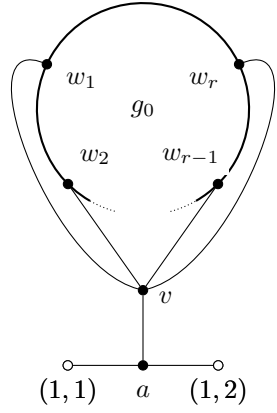
\begin{figure}
\centering
\begin{tikzpicture}[
  label distance=-5.5pt,
  thin,
  vertex/.style={circle,draw=black,fill=black,inner sep=1.25pt,
    minimum size =0mm},
  attach/.style={circle,draw=black,fill=white,inner sep=1.25pt,
    minimum size =0mm},
  dots/.style={circle,fill=black,inner sep=.5pt,
    minimum size= 0pt}]

  \draw (-1,-3.41) -- (1,-3.41);
  
  \foreach \x /\n in {-1/ 1, 1/2}{
    \node at (\x,-3.41) [attach] {};
    \node at (\x,-3.75) {$(1,\n)$};
  }
    
  \node (a) at (0,-3.41) [vertex] {};
  \node at (0,-3.75) {$a$};
  
  \foreach \x /\n in {-1/ 1, 1/2}{
    \node at (\x,-3.41) [attach] {};
    \node at (\x,-3.75) {$(1,\n)$};
  }
  
  \node (v) at (0,-2.41) [vertex]{};
  \node at (0.3,-2.51) {$v$};
  \draw (v) to (a);
  
  \draw[thick] (240:1.41) arc (240:-60:1.41);
  \draw[densely dotted] (240:1.41)  arc (240:255:1.41);
  \draw[densely dotted] (285:1.41)  arc (285:300:1.41);
  
  \foreach \i / \n /\t in {25 / r/ 15, 155 / 1/ 165, 225/ 2 / 120, 315 / {r-1}/60}{
    \node at (\i : .9) [rectangle,fill=white] {$w_\n$};
    \node at (\i : 1.41) [vertex] {};
  }
  \node at (0,0) [rectangle,fill=white] {$g_0$};

  \foreach \i /\t in {25/15, 155/165}{
    \draw (\i:1.41) to[out=\i,in=\t] (v);
  }
  
  \foreach \i \in in {225,315}{
    \draw (\i:1.41) to (v);
  }

\end{tikzpicture}
\caption{A type 2 R/T gadget, i.e., a type 1 gadget with $|S| = 1$.}
\label{fig:reversalRT}
\end{figure}

\begin{figure}
\begin{tikzpicture}[
  label distance=-5.5pt,
  thin,
  vertex/.style={circle,draw=black,fill=black,inner sep=1.25pt,
    minimum size =0mm},
  attach/.style={circle,draw=black,fill=white,inner sep=1.25pt,
    minimum size =0mm},
  dots/.style={circle,fill=black,inner sep=.5pt,
    minimum size= 0pt}]

  \draw (-1,-2.41) -- (1,-2.41);
  
  \foreach \x /\n in {-1/ 1, 1/2}{
    \node at (\x,-2.41) [attach] {};
    \node at (\x,-2.75) {$(1,\n)$};
  }
    
  \node (a) at (0,-2.41) [vertex] {};
  \node at (0,-2.75) {$a$};
  
  \foreach \x /\n in {-1/ 1, 1/2}{
    \node at (\x,-2.41) [attach] {};
    \node at (\x,-2.75) {$(1,\n)$};
  }
  
  \node (v) at (0,2.41) [vertex] {};
  \node at (0,2) {$v$};
  

\foreach \j/\of in {1/0,2/4.82}{
\begin{scope}[yshift=\of cm]
  \foreach \i / \n /\t in {25 / r/ 15, 155 / 1/ 165, 225/ 2 / 120, 315 / {r-1}/60}{
    \node at (\i : .9) [rectangle,fill=white] {$w_\n^{(\j)}$};
    \node at (\i : 1.41) [vertex] {};
  }
  \node at (0,0) [rectangle,fill=white] {$g_0^{(\j)}$};
  
  \draw[thick] (240:1.41) arc (240:-60:1.41);
  \draw[densely dotted] (240:1.41)  arc (240:255:1.41);
  \draw[densely dotted] (285:1.41)  arc (285:300:1.41);

  \foreach \i /\t in {25/15, 155/165}{
    \draw (\i:1.41) to[out=\i,in=\t] (0,-2.41);
  }
  
  \foreach \i \in in {225,315}{
    \draw (\i:1.41) to (0,-2.41);
  }
  
\end{scope}
}

  \draw (25:1.41) to [out = 115,in=-55] (v);
  \draw (155:1.41) to [out=65,in=-125] (v);
  \draw[looseness=1.5] (225:1.41) to[out=180,in=210] (v);
  \draw[looseness=1.5] (315:1.41) to[out=0,in=-30] (v);
\end{tikzpicture}
\caption{The R/T gadget reversing the reflection and transmission sets of \fig{reversalRT}.}
\label{fig:reversal_changed}
\end{figure}
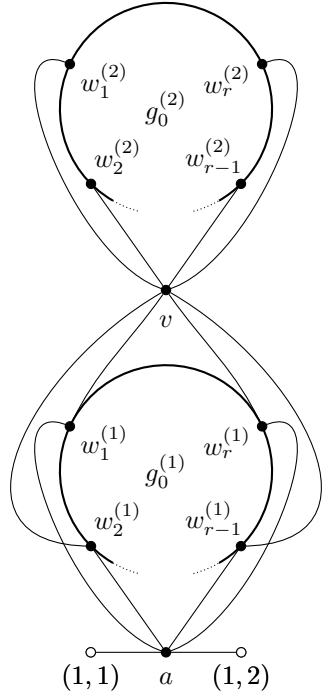

We now show that $\mathcal{R} \subseteq \mathcal{T}'$ and $ \mathcal{T} \subseteq \mathcal{R}'$.  First consider a momentum $k\in \mathcal{T}$. Using the condition derived in \lem{transmit_reqs}, we see that $g_0$ has an eigenvector $\sket{\xi_k}$ with eigenvalue $2\cos(k)$ where the sum of the amplitudes on vertices $w_1,\ldots,w_r$ is nonzero.  Now consider the induced subgraph $G_{0}^{\leftrightarrow}$ of \fig{reversal_changed} obtained by removing vertices $(1,1)$, $(1,2)$, and $a$. This subgraph has an eigenvector $|\chi^{\leftrightarrow}_k\rangle$ with eigenvalue $2\cos(k)$ given by
\[
  \langle w_j^{(i)} \sket{\chi^{\leftrightarrow}_k}= (-1)^{i} \langle{w_j}\sket{\xi_k} \quad \text{ and } \quad \langle{v}\sket{\chi^{\leftrightarrow}_k} = 0.
\]
Since $\sum_{j} \langle{w_j} \sket{\xi_k} \neq 0$, we have $\sum_{j} \langle{w^{(1)}_j} \sket{\chi^{\leftrightarrow}_k} \neq 0$, and using \lem{reflect_reqs} we see that perfect reflection occurs at momentum $k$.  Thus $\mathcal{T} \subseteq \mathcal{R}'$.

Next suppose $k\in \mathcal{R}$. \lem{reflect_reqs} states that $G_0$ has an eigenvector $|\chi_k\rangle$ with eigenvalue $2\cos(k)$ such that $\langle v \sket{\chi_k} \neq 0$.  Now consider the induced subgraph $g_0^{\leftrightarrow}$ of \fig{reversal_changed} obtained by removing vertices $(1,1)$, $(1,2)$, $a$, and $w^{(1)}_1,\ldots,w^{(1)}_n$. This graph has an eigenvector $\sket{\xi^{\leftrightarrow}_k}$ with eigenvalue $2\cos(k)$ defined by
\[
  \langle u \sket{\xi^{\leftrightarrow}_k} = \begin{cases} \langle{u}\sket{\chi_k}& \text{for }u\in g_0^{(2)}\\
    \langle{v}\sket{\chi_k} & u = v\\
0 &\text{otherwise.} \end{cases}
\]
Using this and \lem{transmit_reqs}, we see that $k\in \mathcal{T}'$, so $\mathcal{R} \subseteq \mathcal{T}'$.

\subsection{Examples}
\label{sec:examples}

\subsubsection{Paths}

As a first example, suppose $G_0$ is a finite path and $S$ is one of its vertices. For a path of length $l_1+l_2-2$ (where the length of a path is its number of edges) connected at the $l_1$th vertex as shown in \fig{RT_path}, we determine the reflection and transmission sets as a function of $l_1$ and $l_2$.

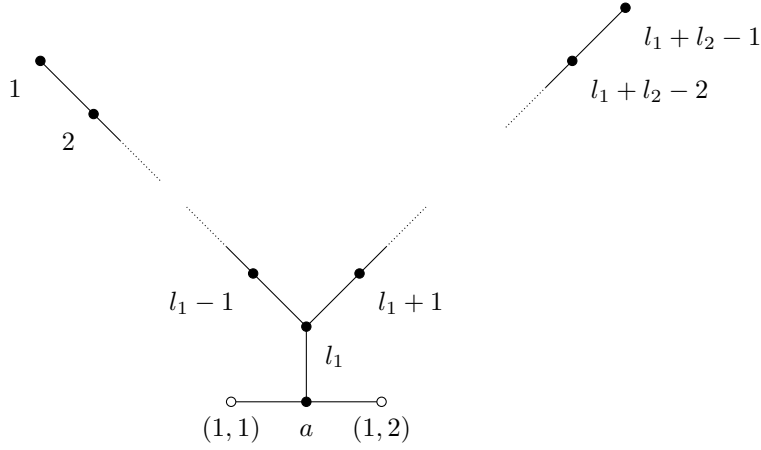
\begin{figure}
\centering
\begin{tikzpicture}[
  thin,
  vertex/.style={circle,draw=black,fill=black,inner sep=1.25pt,
    minimum size =0mm},
  attach/.style={circle,draw=black,fill=white,inner sep=1.25pt,
    minimum size =0mm},
  dots/.style={circle,fill=black,inner sep=.5pt,
    minimum size= 0pt}]
    
    \draw (-1,-1) -- (1,-1);
    \draw (0,0) -- (0,-1);
    
    \node[attach,label=270:{$(1,1)$}] at (-1,-1) {};
    \node[attach,label=270:{$(1,2)$}] at (1,-1) {};
    \node[vertex,label={[label distance=.25*\baselineskip]270:{$ a $}},] at (0,-1) {};
    \node[vertex] at (0,0) {};
    
    \draw (0,0) -- (135:1.5);
    \draw[densely dotted] (135:1.5) -- (135: 2.25);
    
    \draw[densely dotted] (135:2.75) -- (135:3.5);
    \draw (135:3.5) -- (135:5);
    
    \foreach \r in {1,5,4}{
      \node[vertex] at (135:\r) {};
    }

    \foreach \n /\p in {1/5,2/4,{l_1-1}/1}{
      \node at (135:\p) [label={225:$\n$}] {};
    }

    \draw (0,0) -- (45:1.5);
    \draw[densely dotted] (45:1.5) -- (45: 2.25);
    
    \draw[densely dotted] (45:3.75) -- (45:4.5);
    \draw (45:4.5) -- (45:6);
        
    \foreach \r in {1,6,5}{
      \node[vertex] at (45:\r) {};
    }

    \foreach \n /\p in {{l_1}/0,{l_1+1}/1,{l_1+l_2-2}/5,{l_1+l_2-1}/6}{
      \node at (45:\p)[label=315:{$\n$}]{};
    }    

\end{tikzpicture}
\caption{An R/T gadget built from a path of length $l_1+l_2-2$. }
\label{fig:RT_path}
\end{figure}

We use the fact that the path of length $L$ has eigenvectors $|\psi_j\rangle$ for $j\in [L+1]$ given by
\begin{equation}
  \langle x | \psi_j \rangle = \sin\left(\frac{ \pi j x}{L+2}\right)\label{eq:vecs_line}
\end{equation}
with eigenvalues $\lambda_j = 2 \cos(\pi j/(L+2))$. 

Perfect reflection occurs at momentum $k\in (-\pi,0)$ if and only if the path has an eigenvector with eigenvalue $2\cos(k)$ with non-zero amplitude on vertex $l_1$.  Hence
\[
  \mathcal{R}_{\mathrm{path}} = \left\{ -\frac{\pi j}{l_1 + l_2} \colon j\in [l_1 + l_2 - 1] \text{ and } \frac{jl_1}{l_1+l_2} \not\in \ZZ\right\}.
\]

To characterize the momenta at which perfect transmission occurs, consider the induced subgraph obtained by removing the $l_1$th vertex from the path of length $l_1+l_2-2$ (a path of length $l_1-2$ and a path of length $l_2-2$). We can choose bases for the eigenspaces of this induced subgraph so that each eigenvector has all of its support on one of the two paths, and has nonzero amplitude on one of the vertices $l_1-1$ or $l_1+1$. Thus $\hat{G}$ perfectly transmits for all momenta in the set
\[
  \mathcal{T}_{\mathrm{path}} = \left\{- \frac{\pi j}{l_1} \colon j\in [l_1-1]\right\} \cup \left\{-\frac{\pi j}{l_2 } \colon j \in [l_2-1]\right\}.
\]

For example, setting $l_1 = l_2 = 2$, we get $\mathcal{T}_{\mathrm{path}} = \{-\frac{\pi}{2}\}$ and $\mathcal{R}_{\mathrm{path}} = \{-\frac{\pi}{4}, -\frac{3\pi}{4}\}$.

\subsubsection{Cycles}

Suppose $G_0$ is a cycle of length $r$. Labeling the vertices by $x \in [r]$, where $x=r$ is the vertex attached to the path (as shown in \fig{RT_cycle}), the eigenvectors of the $r$-cycle are
\[
  \langle x | \phi_m\rangle = e^{{2 \pi i x m}/{r}}
\]
with eigenvalue $2 \cos(2 \pi m/r)$, where $m\in [r]$. For each momentum $k=-2 \pi m/r \in (-\pi,0)$, there is an eigenvector with nonzero amplitude on the vertex $r$ (i.e., $\langle r | \phi_m\rangle\neq 0$), so \lem{reflect_reqs} implies that perfect reflection occurs at each momentum in the set
\[
  \mathcal{R}_{\mathrm{cycle}} = \left\{ -\frac{\pi j}{r} \colon \text{$j$ is even and $j\in [r-1]$}\right\}.
\]

\begin{figure}
\centering
\begin{tikzpicture}[
  thin,
  vertex/.style={circle,draw=black,fill=black,inner sep=1.25pt,
    minimum size =0mm},
  attach/.style={circle,draw=black,fill=white,inner sep=1.25pt,
    minimum size =0mm},
  dots/.style={circle,fill=black,inner sep=.5pt,
    minimum size= 0pt}]
    \draw (-1,-2) -- (1,-2);
    \draw (0,-2) -- (0,-1);
    
    \node[attach] at (-1,-2) {};
    \node at (-1,-2.35) {$(1,1)$};
    \node[attach] at (1,-2) {};
    \node at (1,-2.35) {$(1,2)$};
    \node[vertex] at (0,-2) {};
    \node at (0,-2.35) {$a$};
    
    \foreach \t in {150, 210, 270, 330, 30}{
      \node[vertex] at (\t: 1) {};
    }
    
    \foreach \t in {210, 270, 330, 30}{
      \draw (\t:1) -- (\t - 60:1);
    }
    
    \draw (150:1) -- (135:.897);
    \draw (30:1) -- (45:.897);
    
    \draw[densely dotted] (45:.897) -- (75:.897);
    \draw[densely dotted] (135:.897) -- (105:.897);
    
    \foreach \t / \n in {210 / 1, 150 / 2, 30 / {r-2}, -30 / {r-1}}{
      \node at (\t:1.5) {$\n$};
    }
    
    \node at (.26,-1.15) {$r$};
  
\end{tikzpicture}
\caption{An R/T gadget built from an $r$-cycle.}
\label{fig:RT_cycle}
\end{figure}
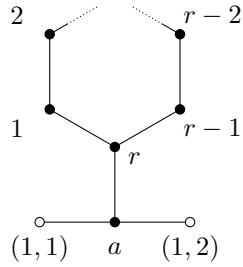

To see which momenta perfectly transmit, consider the induced subgraph obtained by removing vertex $r$. This subgraph is a path of length $r-2$ and has eigenvalues $2\cos(\pi m/r)$ for $m \in [r-1]$ as discussed in the previous section. Using the expression \eq{vecs_line} for the eigenvectors, we see that the sum of the amplitudes on the two ends is nonzero for odd values of $m$.  Perfect transmission occurs for each of the corresponding momenta:  
\[
  \mathcal{T}_{\mathrm{cycle}} = \left\{ -\frac{\pi j}{r} \colon \text{$j$ is odd and $j\in [r-1]$}\right\}.
\]

For example, the $4$-cycle (i.e., square) has $\mathcal{T}_{\mathrm{cycle}} = \{-\frac{\pi}{4},-\frac{3\pi}{4}\}$ and $\mathcal{R}_{\mathrm{cycle}} = \{-\frac{\pi}{2}\}$.  

\section{Constructing momentum switches}\label{sec:switch}

We now construct a momentum switch between the reflection and transmission sets $\mathcal{R}$ and $\mathcal{T}$ of a type 2 R/T gadget.  We attach the gadget and its reversal (defined in \sec{reversal}) to the leaves of a claw, as shown in \fig{gen_mom_con}.

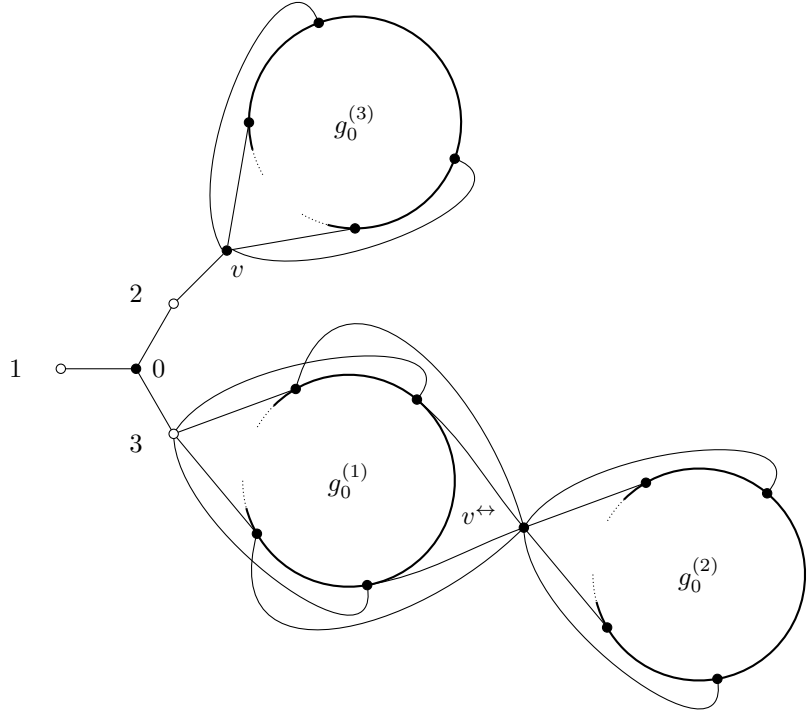
\begin{figure}
\centering
\begin{tikzpicture}[
  label distance=-5.5pt,
  thin,
  vertex/.style={circle,draw=black,fill=black,inner sep=1.25pt,
    minimum size =0mm},
  attach/.style={circle,draw=black,fill=white,inner sep=1.25pt,
    minimum size =0mm},
  dots/.style={circle,fill=black,inner sep=.5pt,
    minimum size= 0pt}]

\begin{scope}[yshift=-.866 cm,xshift=.5cm,rotate=255,yshift = 2.41cm]

  
  \node (v) at (0,2.41) [vertex] {};
  \node at (0,1.8) {$v^\leftrightarrow$};
  

\foreach \j/\of in {1/0,2/4.82}{
\begin{scope}[yshift=\of cm]
  \foreach \i / \n /\t in {25 / n/ 15, 155 / 1/ 165, 225/ 2 / 120, 315 / {n-1}/60}{

    \node at (\i : 1.41) [vertex] {};
  }
  \node at (0,0) [rectangle,fill=white] {$g_0^{(\j)}$};
  
  \draw[thick] (240:1.41) arc (240:-60:1.41);
  \draw[densely dotted] (240:1.41)  arc (240:255:1.41);
  \draw[densely dotted] (285:1.41)  arc (285:300:1.41);
  
  \foreach \i /\t in {25/15, 155/165}{
    \draw (\i:1.41) to[out=\i,in=\t] (0,-2.41);
  }
  
  \foreach \i \in in {225,315}{
    \draw (\i:1.41) to (0,-2.41);
  }
  
\end{scope}
}

  \draw (25:1.41) to [out = 115,in=-55] (0,2.41);
  \draw (155:1.41) to [out=65,in=-125] (0,2.41);
  \draw[looseness=1.5] (225:1.41) to [out=180,in=210] (0,2.41);
  \draw[looseness=1.5] (315:1.41) to [out=0,in=-30] (0,2.41);

\end{scope}

\begin{scope}[xshift = .5 cm, yshift = .866 cm, rotate = -45, yshift = 3.41 cm]
  
  \node (v) at (0,-2.41) [vertex]{};
  \node at (0.3,-2.51) {$v$};
  \draw (v) -- (0,-3.41); 
  
  \foreach \i / \n /\t in {25 / n/ 15, 155 / 1/ 165, 225/ 2 / 120, 315 / {n-1}/60}{

    \node at (\i : 1.41) [vertex] {};
  }
  \node at (0,0) [rectangle,fill=white] {$g_0^{(3)}$};
  
  \draw[thick] (240:1.41) arc (240:-60:1.41);
  \draw[densely dotted] (240:1.41)  arc (240:255:1.41);
  \draw[densely dotted] (285:1.41)  arc (285:300:1.41);

  \foreach \i /\t in {25/15, 155/165}{
    \draw (\i:1.41) to[out=\i,in=\t] (v);
  }
  
  \foreach \i \in in {225,315}{
    \draw (\i:1.41) to (v);
  }
  
\end{scope}

\foreach \t in {60, 180, 300}{
  \draw (0,0) -- (\t:1);
}

\node (0) at (0,0) [vertex] {};
\node (1) at (180:1) [attach] {};
\node (2) at (60:1) [attach]{};
\node (3) at (-60:1) [attach]{};

\node at (0:.3) {$0$};
\node at (-1.6,0) {$1$};
\node at (0,1) {$2$};
\node at (0,-1) {$3$};
\end{tikzpicture}
\caption{A momentum switch built from a type 2 R/T gadget and its reversal.}
\label{fig:gen_mom_con}
\end{figure}

Recall that for each $k\in \mathcal{T}$, the graph $g_0$ has a $2\cos(k)$-eigenvector $\sket{\xi_k}$ satisfying equation \eq{trans_cond} with some nonzero constant $c$. We define a state $\sket{\mu_k}$ on the infinite graph obtained by attaching three semi-infinite paths to the gadget shown in  \fig{gen_mom_con} and we show that it is a scattering eigenstate with perfect transmission between paths $1$ and $2$.   The amplitudes of $\sket{\mu_k}$ on the semi-infinite paths and the claw are given by
\[
  \langle (x,1)|\mu_k\rangle=e^{-ikx} \qquad 
  \langle 0|\mu_k\rangle=1 \qquad 
  \langle (x,2)|\mu_k\rangle=e^{ikx} \qquad
  \langle (x,3)|\mu_k\rangle=0.
\]
The rest of the graph consists of the three copies of the subgraph $g_0$ and the vertices $v$ and $v^{\leftrightarrow}$. The corresponding amplitudes are
\[
  \sbraket{u}{\mu_k} = \begin{cases} 	-\frac{1}{c}\sbraket{u}{\xi_k} &  u\in g_0^{(1)}\\
           \frac{1}{c}\sbraket{u}{\xi_k} &  u\in g_0^{(2)}\\
	-\frac{e^{ik}}{c} \sbraket{u}{\xi_k} & u\in g_0^{(3)}\\
  	0 & u=v \text{ or } u=v^{\leftrightarrow}.
          \end{cases}
\]

One can check that this is an eigenvector with eigenvalue $2\cos(k)$. We see that each momentum $k\in \mathcal{T}$ perfectly transmits from path 1 to path 2.

Similarly, one can construct an eigenstate with perfect transmission from path 1 to path 3 for each momentum $p\in \mathcal{R}$. This shows that the graph gadget from \fig{gen_mom_con} is a momentum switch between $\mathcal{R}$ and $\mathcal{T}$, as claimed.

Using this construction, we can obtain a momentum switch from any of the examples discussed in \sec{examples}.  For example, using the R/T gadget built from the 3-cycle, we get a momentum switch between $-\frac{\pi}{3}$ and $-\frac{2\pi}{3}$, as shown in \fig{mom_switch_ex}.  More generally, using an $r$-cycle, we obtain a switch between momenta of the form $-\frac{\pi j}{r}$ with odd or even values of $j$.  As another example, using a path of length $4$ connected at the center vertex, we obtain a switch between $-\frac{\pi}{4}$ and $-\frac{\pi}{2}$ that differs from the one shown in \fig{mswitch}.

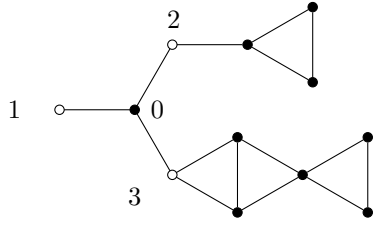
\begin{figure}
\begin{tikzpicture}[
  label distance=-5.5pt,
  thin,
  vertex/.style={circle,draw=black,fill=black,inner sep=1.25pt,
    minimum size =0mm},
  attach/.style={circle,draw=black,fill=white,inner sep=1.25pt,
    minimum size =0mm},
  dots/.style={circle,fill=black,inner sep=.5pt,
    minimum size= 0pt}]

\begin{scope}[yshift = .866cm, xshift = .5cm,rotate=0]

  \node (b) at (1,0) [vertex] {};
  \node (c) at (1.866,.5) [vertex] {};
  \node (d) at (1.866,-.5) [vertex] {};
  
  \draw (0,0) -- (b) -- (c) -- (d) -- (b);
\end{scope}

\begin{scope}[yshift = -.866cm, xshift = .5cm,rotate= -90]
  \node (e) at (.5,.866) [vertex] {};
  \node (f) at (-.5,.866) [vertex] {};
  \node (g) at (0,1.732) [vertex] {};
  \node (h) at (.5,2.598) [vertex] {};
  \node (i) at (-.5,2.598) [vertex]{};
  
  \draw (0,0) -- (e) -- (g) -- (h) -- (i) -- (g) -- (f) -- (0,0);
  \draw (e) -- (f);

\end{scope}

\foreach \t in {60, 180, 300}{
  \draw (0,0) -- (\t:1);
}

\node (0) at (0,0) [vertex] {};
\node (1) at (180:1) [attach] {};
\node (2) at (60:1) [attach]{};
\node (3) at (-60:1) [attach]{};

\node at (0:.3) {$0$};
\node at (-1.6,0) {$1$};
\node at (0.52,1.2) {$2$};
\node at (0,-1.15) {$3$};

\end{tikzpicture}
\caption{A momentum switch between $-\frac{\pi}{3}$ and $-\frac{2\pi}{3}$.}
\label{fig:mom_switch_ex}
\end{figure}

\section{Impossibility of a momentum switch between $-\frac{\pi}{4}$ and $-\frac{3\pi}{4}$}
\label{sec:impossibility}

In this Section we prove that there does not exist a momentum switch between momenta $-\frac{\pi}{4}$ and $-\frac{3\pi}{4}$. We begin by proving that there is a basis for the space of scattering states with momentum $k=-\frac{\pi}{4}$ or $k=-\frac{3\pi}{4}$ where each basis vector has entries in $\QQ(\sqrt{2})$. We then use this fact to prove that there is no R/T gadget between these two momenta. Since any momentum switch can be converted into an R/T gadget between the momenta it separated (as shown in \sec{mswitch}), this implies that no momentum switch exists between $-\frac{\pi}{4}$ and $-\frac{3\pi}{4}$.

\subsection{Basis vectors with entries in $\QQ(\sqrt{2})$}
\label{sec:vecs_over_field}

Recall the general setup shown in \fig{basic_scattering}: $N$ semi-infinite paths are attached to a finite graph $\hat G$. Consider an eigenvector $\sket{\tau_k}$ of the adjacency matrix of $G$ with eigenvalue $2\cos(k)$ for $k\in (-\pi,0)$. In general this eigenspace is spanned by scattering states with momentum $k$ and so-called confined bound states \cite{CG12} (which have zero amplitude on the semi-infinite paths). We can write the amplitudes of $\sket{\tau_k}$ on the semi-infinite paths as
\[
  \sbraket{(x,j)}{\tau_k} 
  = \kappa_j \cos(k (x-1)) + \sigma_j \sin(k (x-1))
\]
for $x \in \posint$, $j \in [N]$, and $\kappa_j,\sigma_j \in \CC$, and the amplitudes on the internal vertices as
\[
  \sbraket{w}{\tau_k} = \iota_w
\]
for $\iota_w \in \CC$, where $w$ indexes the internal vertices. Write the adjacency matrix of $\hat{G}$ as a block matrix
\[
  A(\hat{G}) = \begin{pmatrix} A & B\\ B^\dag & D\end{pmatrix}
\]
where the first block corresponds to the vertices attached to semi-infinite paths and the second block corresponds to the internal vertices.  The eigenvalue equation for $\sket{\tau_k}$ can be written 
\begin{equation*}
  \begin{pmatrix} A & B\\ B^\dag & D\end{pmatrix}
	\begin{pmatrix} \kappa \\ \iota \end{pmatrix}
	+ \cos(k) \begin{pmatrix} \kappa \\ 0 \end{pmatrix}
	+ \sin(k) \begin{pmatrix} \sigma \\ 0 \end{pmatrix} 
	= 2\cos(k) \begin{pmatrix} \kappa \\ \iota \end{pmatrix},
\end{equation*}
so the nullspace of the matrix
\[
  M= \begin{pmatrix} A -\cos(k) \II & \sin(k) \II & B\\
    0 & 0 & 0\\
    B^\dag & 0 & D-2\cos(k)\II\end{pmatrix}
\]
is in one-to-one correspondence with the $2\cos(k)$-eigenspace of the infinite matrix. Further, this matrix only has entries in $\QQ(\cos(k),\sin(k))$, so its nullspace has a basis with amplitudes in $\QQ(\cos(k),\sin(k))$, as can be seen using Gaussian elimination.

We are interested in the specific cases $2\cos(k)=\pm\sqrt{2}$ corresponding to $k = -\frac{\pi}{4}$ or $k = -\frac{3\pi}{4}$.  In these cases $\QQ(\cos(k),\sin(k)) = \QQ(\sqrt{2})$, and we may choose a basis for the nullspace of $M$ with amplitudes from $\QQ(\sqrt{2})$. Furthermore, $\cos(kx), \sin(kx) \in \QQ(\sqrt{2})$ for all $x\in \posint$, so with an appropriate choice of basis, each amplitude of $\sket{\tau_k}$ is also an element of $\QQ(\sqrt{2})$.

As noted above, the spectrum of $G$ may include confined bound states \cite{CG12} with eigenvalue $\pm\sqrt2$.  However, any such states are eigenstates of $A(\hat{G})$ subject to the additional (rational) constraints that the amplitudes on the vertices connected to the semi-infinite paths are zero.  As such, the confined bound states have a basis over $\QQ(\sqrt{2})$. Likewise there exists a basis over $\QQ(\sqrt{2})$ for the subspace of \emph{scattering states} with energy $\pm\sqrt{2}$, i.e., the $N$-dimensional space orthogonal to the confined bound states. Finally, note that for any member of this basis $\sket{\tau_k}$ there exist rational vectors $\sket{u_k},\sket{w_k}$ such that $\sket{\tau_k}=\sket{u_k} + \sqrt{2}\sket{w_k}$. Since $H^2\sket{\tau_k}=2\sket{\tau_k}$, we have $H\sket{u_k}=\pm 2\sket{w_k}$ and $H\sket{w_k} = \pm \sket{u_k}$, so
\begin{equation}
  \sket{\tau_k}=(H \pm \sqrt{2} \II)\sket{w_k}.\label{eq:form_of_tau}
\end{equation}

\subsection{No R/T gadget and hence no momentum switch}

Recall from \sec{mswitch} that a momentum switch between two momenta $k$ and $p$ can always be converted into an R/T gadget between $k$ and $p$. Here we show that if a graph perfectly reflects at momentum $-\frac{\pi}{4}$, then it must also perfectly reflect at momentum $-\frac{3\pi}{4}$. This  implies that no R/T gadget exists between these two momenta, and thus no momentum switch exists.

We use the following basic fact about two-terminal gadgets: if a state $\sket{\phi}$ within the span of the momentum-$k$ scattering states (i.e., any $2\cos(k)$-eigenstate orthogonal to the confined bound states) has zero amplitude along one of the paths, then it is a scalar multiple of one of the scattering eigenstates and the gadget perfectly reflects at momentum $k$.  This holds because if $\sket{\phi}$ has zero amplitude along, say, path $2$, then there exist some $\mu,\nu \in \CC$ such that
\[
  \sbraket{(x,2)}{\phi} 
  = \mu \sbraket{(x,2)}{\sc_2 (k)} + \nu \sbraket{(x,2)}{\sc_1(k)} 
  = \mu e^{-ikx} + \mu R e^{ikx} + \nu T e^{ikx} 
  = 0
\] 
for all $x\in \ZZ^{+}$.  Since this holds for all $x$, we have $\mu = \mu R + \nu T = 0$.  Since $\mu$ and $\nu$ cannot both be zero, we have $T=0$, and $\sket{\phi} \propto \sket{\sc_1(k)}$.

For an R/T gadget, the scattering states at a fixed momentum span a two-dimensional space. As shown in \sec{vecs_over_field}, we can expand each scattering eigenstate at momentum $k=-\frac{\pi}{4}$ in a basis with entries in $\QQ(\sqrt{2})$, where each basis vector takes the form \eq{form_of_tau}. This gives
\begin{equation*}
  \sket{\sc_{1}(-\tfrac{\pi}{4})} 
  = (H + \sqrt{2}\II) (\alpha \sket{a} + \beta \sket{b}) \label{eq:rat_expansion}
\end{equation*}
where $\alpha,\beta \in \CC$, $\alpha\neq 0$, and $\sket{a}$ and $\sket{b}$ are rational vectors.

If $T(-\frac{\pi}{4}) = 0$, then for all $x \geq 0$, 
\[
  \langle x,2\sket{\sc_1(-\tfrac{\pi}{4})} 
  = 0 
  = \langle{x,2} |(H + \sqrt{2}\II) (\alpha \sket{a} + \beta \sket{b}).
\]
Dividing through by $\alpha$ and rearranging, we get that for all $x\geq 0$,
\begin{equation*}
\frac{\beta}{\alpha} (\langle{x,2}| H\sket{b}+\sqrt{2} \langle{x,2}\sket{b})
  =-\langle{x,2}| H\sket{a}  -\sqrt{2} \langle{x,2}\sket{a}.
\label{eq:cases_eqn}
\end{equation*}
If the left-hand side is not zero, then $\beta/\alpha \in \QQ(\sqrt{2})$.  If the left-hand side is zero, then $(H+ \sqrt{2}\II)\sket{a}$ is an eigenstate at energy $2\cos(k)$ with no amplitude along path 2, so $\beta = 0$ (using the fact about two-terminal gadgets), and again $\beta/\alpha \in \QQ(\sqrt{2})$.

Now write $\beta/\alpha=r+s\sqrt{2}$ with $r,s\in\QQ$, and consider the rational vector
\[
  \sket{c} := \sket{a} + (r+sH) \sket{b}.
\]
Note that
\[
  \alpha (H + \sqrt{2} \II) \sket{c} 
= \alpha(H+ \sqrt{2} \II) \sket{a} + \alpha (rH+r\sqrt{2}+ sH^2+sH\sqrt{2}) \sket{b}.
\]
Since $\sket{b}$ is a 2-eigenvector of $H^2$ and $\beta/\alpha=r+s\sqrt{2}$, this simplifies to
\begin{equation}
  \alpha (H + \sqrt{2}\II) \sket{c} 
  = \alpha (H + \sqrt{2}\II) \sket{a} + \beta(H + \sqrt{2}\II) \sket{b} 
  = \sket{\sc_1(-\tfrac{\pi}{4})}, \label{eq:sc1_c}
\end{equation}
so $\sket{\sc_1(-\tfrac{\pi}{4})}$ can be written as $\alpha(H+\sqrt{2}\II)$ times a rational 2-eigenvector of $H^2$.

Since $\langle{x,2}\sket{\sc_1(-\tfrac{\pi}{4})} = 0$ for all $x\geq 0$ (and $\alpha\neq 0$), we have
\[
  \langle{x,2}| (H+\sqrt{2}\II) \sket{c} 
  = \langle{x,2}| H \sket{c} + \sqrt{2} \langle{x,2}\sket{c} 
  = 0.
\]
As $H$ is a rational matrix and $\sket{c}$ is a rational vector, the rational and irrational components must both be zero, implying $\langle{x,2}\sket{c}  =\langle{x,2}|H\sket{c} = 0$ for all $x\geq 0$. Furthermore, since $ \sket{\sc_1(-\tfrac{\pi}{4})}$ is a scattering state with zero amplitude on path $2$, it must have some nonzero amplitude on path 1 and thus there is some $x_0\in \mathbb{Z}^+$ for which $\langle{x_0,1}\sket{c}\neq 0$ or $\langle{x_0,1}|H\sket{c} \neq 0$.

Now consider the state obtained by replacing $\sqrt{2}$ with $-\sqrt{2}$:
\[
  \sket{\overline{\sc}_1(-\tfrac{\pi}{4})} := \alpha (H-\sqrt{2}\II) \sket{c}.
\]
This is a $-\sqrt{2}$-eigenvector of $H$, which can be confirmed using the fact that $\sket{c}$ is a $2$-eigenvector of $H^2$. As $\langle{x,2} | H \sket{c} = \langle x,2 \sket{c} = 0$ for all $x\geq 0$, $\langle x,2 \sket{\overline{\sc}_1(-\tfrac{\pi}{4})} = 0$ for all $x\geq 0$. Furthermore the amplitude at vertex $(x_0,1)$ is nonzero, i.e.,  $\langle x_0,1 \sket{\overline{\sc}_1(-\tfrac{\pi}{4})} \neq 0$, and hence $\sket{\overline{\sc}_1(-\tfrac{\pi}{4})}$ has a component orthogonal to the space of confined bound states (which have zero amplitude on both semi-infinite paths).  Hence, there exists a scattering state with eigenvalue $-\sqrt{2}$ with no amplitude on path 2.  This scattering state must be a multiple of $\sket{\sc_1(-\frac{3\pi}{4})}$, so $\sket{\sc_1(-\frac{3\pi}{4})}$ perfectly reflects (using the fact about two-terminal gadgets).  Hence, perfect reflection at momentum $-\frac{\pi}{4}$ implies perfect reflection at momentum $-\frac{3\pi}{4}$.  It follows that no perfect R/T gadget (and hence no perfect momentum switch) exists between these momenta.

This proof technique can also establish non-existence of momentum switches between other pairs of momenta $k$ and $p$.  For example, a slight modification of the above proof shows that no momentum switch exists between $k = -\frac{\pi}{6}$ and $p = -\frac{5\pi}{6}$. 

\section{An approximate switch between $-\frac{\pi}{4}$ and $-\frac{3\pi}{4}$}
\label{sec:approx_switch}

Although no perfect switch exists between momenta $-\frac{\pi}{4}$ and $-\frac{3\pi}{4}$, in this Section we construct a sequence of graphs that approximates such a switch arbitrarily well.  The switch works by splitting the wave packet into two pieces, applying a momentum-dependent relative phase, and recombining the pieces.

The first ingredient in our construction is the graph $G_\bc$ shown in \fig{basis_change}, which was used in the single- and multi-particle universality constructions to implement a basis-changing gate \cite{Chi09,CGW13}. At momenta $k=-\frac{\pi}{4}$ and $-\frac{3\pi}{4}$, the S-matrix of this graph has the form
\begin{equation*}
  S(-\tfrac{\pi}{4}) = \begin{pmatrix}
    0 & U_\bc\\
    U_\bc& 0\end{pmatrix} \qquad
 S(-\tfrac{3\pi}{4}) = \begin{pmatrix}
    0 & -U^{\ast}_\bc\\
    -U^{\ast}_\bc& 0\end{pmatrix},
\label{eq:S_basis}
\end{equation*}
where each block has size $2\times 2$ and
\begin{equation*}
  U_\bc= -\frac{1}{\sqrt{2}} \begin{pmatrix} i & 1\\ 1 & i\end{pmatrix}.
\label{eq:U_basis}
\end{equation*}

\begin{figure}
\centering
\begin{tikzpicture}
  [ scale=0.7,
    yscale=1.5,
    inner/.style={circle,draw=black!100,fill=black!100,inner sep = 1.25pt},
    attach/.style={circle,draw=black!100,fill=black!0,thin,inner sep = 1.25pt}]

  \node (1) at ( 0, 2) [inner]  {};
  \node (2) at ( 0, 0) [inner]  {};
  \node (3) at ( 1, 2) [inner]  {};
  \node (4) at ( 1, 0) [inner]  {};
  \node (5) at ( 0, 1) [inner]  {};
  \node (6) at ( 1, 1) [inner]  {};
  \node (7) at (-1, 2) [attach,label=left:{$(1,1)$}] {};
  \node (8) at (-1, 0) [attach,label=left:{$(1,2)$}] {};
  \node (9) at ( 2, 2) [attach,label=right:{$(1,3)$}] {};
  \node (0) at ( 2, 0) [attach,label=right:{$(1,4)$}] {};

  \draw (7) to (1) [thin];
  \draw (8) to (2) [thin];
  \draw (3) to (9) [thin];
  \draw (4) to (0) [thin];
  \draw (1) to (3) [thin];
  \draw (1) to (5) [thin];
  \draw (2) to (4) [thin];
  \draw (2) to (5) [thin];
  \draw (3) to (6) [thin];
  \draw (4) to (6) [thin];
\end{tikzpicture}
\caption{A graph $G_\bc$ that implements a basis-changing gate at $-\frac{\pi}{4}$ and $-\frac{3\pi}{4}$.}
\label{fig:basis_change}
\end{figure}
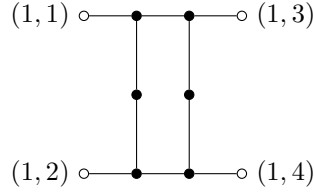 

The second ingredient, which we use to apply a momentum-dependent phase,  is the graph $G_\ph$ shown in \fig{approx_base}. This graph has perfect transmission at both momenta of interest, with transmission coefficients $T(-\tfrac{\pi}{4}) = -e^{i\phi}$ and $T(-\tfrac{3\pi}{4}) = e^{i\phi}$, where 
\begin{equation}
  e^{i\phi} =\frac{2 \sqrt{2}}{3} + \frac{i}{3} = e^{i \arctan\frac{1}{2\sqrt{2}}}.
\label{eq:phi}
\end{equation}

We construct an approximate momentum switch from $G_\ph$ and $G_\bc$ as shown in \fig{approx_switch}.  Here we use $m$ copies of $G_\ph$ for some odd $m$ that depends on the precision required in the approximation. To understand the scattering matrix of this graph, we use the following fact. Suppose graphs $G_1$ and $G_2$ each have one input terminal and one output terminal, and both have perfect transmission at some fixed momentum $k$, i.e., $|T_1(k)|=|T_2(k)|=1$. Now consider the gadget obtained by merging $G_1$ with $G_2$ by identifying the output vertex of $G_1$ with the input vertex of $G_2$ (now the input terminal is that of $G_1$ and the output terminal is that of $G_2$). Then the resulting graph has perfect transmission with transmission coefficient $e^{2ik}T_1(k)T_2(k)$. 

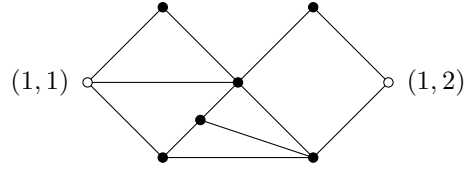
\begin{figure}
\centering
\begin{tikzpicture}[
  label distance=-5.5pt,
  thin,
  vertex/.style={circle,draw=black,fill=black,inner sep=1.25pt,
    minimum size =0mm},
  attach/.style={circle,draw=black,fill=white,inner sep=1.25pt,
    minimum size =0mm},
  dots/.style={circle,fill=black,inner sep=.5pt,
    minimum size= 0pt}]
    
\foreach \i in {1,-1}
	\draw (0,0)--(1*\i,1)--(2*\i,0)--(1*\i,-1)--(0,0);
\draw  (-2,0)--(0,0);
\draw  (-1,-1)--(1,-1);
\draw (-0.5,-0.5)--(1,-1);

\node at (0,0) [vertex] {};
\foreach \i in {1,-1}
	\foreach \j in {1,-1}
		\node at (\i,\j) [vertex]{};
\node at (-0.5,-0.5) [vertex] {};

\node at (-2,0) [attach,label={[label distance = .05 cm]left:$(1,1)$}] {};

\node at (2,0) [attach,label={[label distance = .05 cm]right:$(1,2)$}] {};

\end{tikzpicture}
\caption{A graph $G_\ph$ with perfect transmission and irrational argument at $-\frac{\pi}{4}$ and $-\frac{3\pi}{4}$.}
\label{fig:approx_base}
\end{figure}

\begin{figure}
\centering
\begin{tikzpicture}
  [ thin,
    inner/.style={circle,draw=black!100,fill=black!100,inner sep = 1.25pt},
    attach/.style={circle,draw=black!100,fill=black!0,thin,inner sep = 1.25pt},
    used/.style={circle,draw=black!100,fill=black!100,thin,inner sep = 1.25pt},
    dots/.style={circle,draw=black!100,fill=black!100,thin,inner sep = .4pt}]


\draw (2,0) -- (10,0);

\foreach \x in {5.5}{
  \node at (\x,0) [inner] {};
}


\foreach \x in {0, 10}{
\begin{scope}[xshift = \x cm]
  \node (1) at ( 0, 2) [inner]  {};
  \node (2) at ( 0, 0) [inner]  {};
  \node (3) at ( 1, 2) [inner]  {};
  \node (4) at ( 1, 0) [inner]  {};
  \node (5) at ( 0, 1) [inner]  {};
  \node (6) at ( 1, 1) [inner]  {};
  \node (a) at (-1,2) [attach] {};
  \node (c) at (2,2) [attach]{};
  \node (d) at (2,0) [attach]{};
  \node (e) at (-1,0) [attach]{};

  \draw (1) to (3);
  \draw (1) to (5);
  \draw (2) to (4);
  \draw (2) to (5);
  \draw (3) to (6);
  \draw (4) to (6);
  \draw (a) to (1);
  \draw (c) to (3);
  \draw (d) to (4);
  \draw (e) to (2);
\end{scope}}

\foreach \y in {0,2}{
\foreach \x in {2,11}{
  \node at (\x,\y) [used]{};
}}

  \foreach \x in {2,3.5,7.5}{
    \draw (\x + .75,2) [thick] circle (.75 cm);
    \node at (\x + .75,2) {$G_{\text{ph}}$};
    \node at (\x,2) [used] {};
    \node at (\x+1.5,2) [used] {};
  }
  
  \draw[thick,shift=(150:.75cm)](5.75,2) arc (150:210:.75cm);
  \draw[densely dotted] ([shift=(155:.75cm)](5.75,2) arc (155:130:.75cm);
  \draw[densely dotted] ([shift=(205:.75cm)](5.75,2) arc (205:230:.75cm);
  
  \draw[thick,shift=(30:.75cm)](6.75,2) arc (30:-30:.75cm);
  \draw[densely dotted] ([shift=(25:.75cm)](6.75,2) arc (25:50:.75cm);
  \draw[densely dotted] ([shift=(335:.75cm)](6.75,2) arc (335:310:.75cm);
  
  \foreach \x in {6.05,6.25,6.45}{
    \node[dots] at (\x,2){};
  }
  
\draw [decorate,decoration={brace,amplitude=5pt}] (2,2.9) -- (9,2.9) 
	node [midway,above=4pt] {$m$ copies};
	

\node at (-1,2) [label=below:{$(1,1)$}] {};
\node at (-1,0) [label=below:{$(1,2)$}] {};
\node at (12,2) [label=below:{$(1,3)$}] {};
\node at (12,0) [label=below:{$(1,4)$}] {};

\node at (2,2) [anchor=north east,label={$a$}] {};
\node at (9.35,2.05) [anchor=south west,label=left:{$c$}] {};
\node at (2,0) [label=below:{$b$}] {};
\node at (9,0) [inner,label=below:{$d$}] {};

\draw[loosely dashed] (1.5,-.8) -- (9.5,-.8) -- (9.5,3.75) -- (1.5,3.75) -- (1.5,-.8);

\end{tikzpicture}
\caption{An approximate momentum switch between $-\frac{\pi}{4}$ and $-\frac{3\pi}{4}$ (for suitable values of $m$).}
\label{fig:approx_switch}
\end{figure}
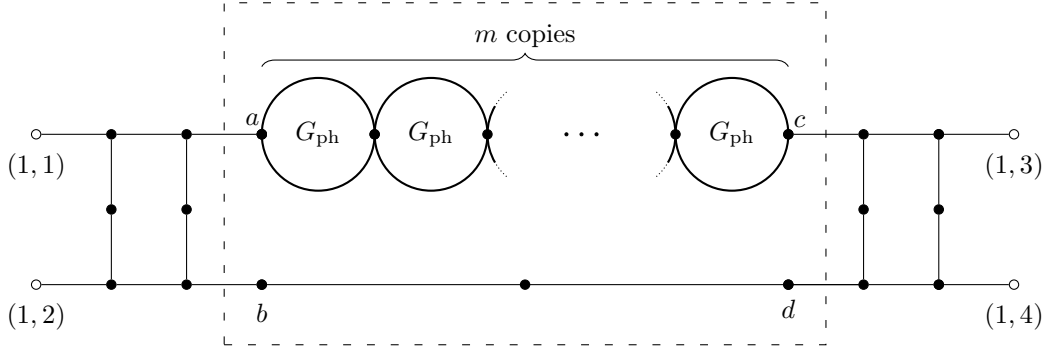

Using this fact and equation \eq{phi} we see that the graph obtained by merging $m$ copies of $G_\ph$ in this way (with $m$ odd) has transmission coefficients $T(-\tfrac{3\pi}{4})=-T(-\tfrac{\pi}{4}) = i^{(m-1)} e^{im\phi}$. Now look at the induced subgraph of \fig{approx_switch} on vertices contained within the dotted box, and consider attaching semi-infinite paths to vertices labeled $a,b,c,d$. Using the fact that the path with two edges has perfect transmission with coefficient $1$ (at any momentum), we see that the S-matrix of this gadget is
\[
\begin{pmatrix}
    0 & U_m(k)\\
    U_m(k)& 0\end{pmatrix}
\]
where 
\[
U_m(-\tfrac{\pi}{4})=\begin{pmatrix}
    - i^{(m-1)} e^{im\phi} & 0\\
    0& 1\end{pmatrix}\qquad
U_m(-\tfrac{3\pi}{4})=\begin{pmatrix}
     i^{(m-1)} e^{im\phi} & 0\\
    0& 1\end{pmatrix}.
\]

The full graph shown in \fig{approx_switch} is obtained from this subgraph by merging it with two copies of $G_{\mathrm{bc}}$ in a similar way to the merging procedure described above. Since each of the subgraphs being merged has perfect transmission from input terminals (on the left) to output terminals (on the right), their S-matrices compose in a simple way. At both momenta $k\in \{-\frac{\pi}{4},-\frac{3\pi}{4}\}$, the overall S-matrix has perfect transmission from the input paths on the left-hand side to the output paths on the right-hand side, and takes the form
\begin{equation*}
\label{eq:Smat1}
\begin{pmatrix}
    0 & V(k)\\
    V(k)& 0\end{pmatrix}
\end{equation*}
where
\begin{align*}
V(-\tfrac{\pi}{4})& =-U_\bc U_m(-\tfrac{\pi}{4})U_\bc=-\frac{1}{2} \begin{pmatrix}
    -i^{m+1}e^{im\phi}+1 & -i^{m}e^{im\phi}+i\\
    -i^{m}e^{im\phi}+i & -i^{m-1}e^{im\phi}-1 \end{pmatrix}\\
V(-\tfrac{3\pi}{4})& =-U_\bc^{\ast} U_m(-\tfrac{3\pi}{4})U_\bc^{\ast}=-\frac{1}{2} \begin{pmatrix}
    i^{m+1}e^{im\phi}+1 &- i^{m}e^{im\phi}-i\\
    -i^{m}e^{im\phi}-i & i^{m-1}e^{im\phi}-1 \end{pmatrix}
\end{align*}
assuming $m$ is odd. Since $i^{m+1}=-i^{m-1}= \pm 1$, we can see from these expressions that if either $e^{im\phi} \approx 1$ or $e^{im\phi} \approx -1$, then the graph is close to a momentum switch at these momenta. (More precisely, since a momentum switch is a three-terminal gadget and this graph has four terminals, we obtain an approximate momentum switch from this graph by downgrading the terminal vertex $(1,2)$ to an internal vertex). Since $\arctan(2^{-3/2})$ is an irrational multiple of $\pi$, the set $\{e^{2i j \phi}\colon j \in \ZZ^+\}$ is dense on the unit circle, so for any $\epsilon>0$  and choice of sign $\pm$, there exists some $j \in \ZZ^+$ such that $|e^{i(2j+1)\phi} \pm 1| = |e^{2ij\phi} \pm e^{-i\phi}| < \epsilon$. Taking $m=2j+1$ copies of $G_{\mathrm{ph}}$, this lets us approximate a momentum switch between $-\frac{\pi}{4}$ and $-\frac{3\pi}{4}$ to any desired precision.  In particular, $m=37$ gives an approximation with
\[
  \left\|V (-\tfrac{\pi}{4}) - \begin{pmatrix} -1 & 0\\ 0 & 1\end{pmatrix} \right\| \approx 0.0076 \approx \left\| V (-\tfrac{3\pi}{4}) - \begin{pmatrix} 0 & i \\ i & 0\end{pmatrix}\right\|.
\]
The next value of $m$ yielding a better approximation is $m=379$, with an error of approximately $0.0071$.

\section{Discussion}
\label{sec:conc}

In this work we have constructed momentum switches that route a quantum walker along a path that depends on its momentum.  Our results could  be used to design variants of the multi-particle quantum walk universality construction that use qubits encoded as particles with different momenta (the original construction \cite{CGW13} used momenta  $-\frac{\pi}{4}$ and $-\frac{\pi}{2}$). More broadly, we hope that tools for designing scattering gadgets will be useful for developing new quantum algorithms based on continuous-time quantum walk.

We also gave an example showing that (perfect) momentum switches cannot always be constructed.  Exact implementation of an S-matrix by scattering on an unweighted graph is analogous to exact synthesis of unitary operations using a finite set of gates \cite{KMM13,GS13}.  It might be interesting to further explore the set of S-matrices that can be realized by scattering on graphs, and perhaps to characterize the set of momentum switches that can be implemented.

Other avenues for research also remain open. Many of our results only apply to graphs in a restricted family.  In particular, our understanding of R/T gadgets is mostly limited to those of type 1 (although our result concerning non-existence of an R/T gadget between momenta $-\frac{\pi}{4}$ and $-\frac{3\pi}{4}$ is more general). It would be more satisfying to determine necessary and sufficient conditions for a graph to be an R/T gadget (or a momentum switch) without restricting the form of the gadget.

More generally, one might consider the problem of designing scattering gadgets with other restrictions on the allowed Hamiltonian.  Here we have assumed that the Hamiltonian is the adjacency matrix of a simple graph.  One might also consider, say, Laplacians of graphs.  Another natural model would allow matrices whose entries are unrestricted, but that can have at most some number of nonzero entries in each row (i.e., whose underlying graphs have bounded degree).

\section*{Acknowledgments}

This work was supported in part by NSERC; the Ontario Ministry of Research and Innovation; the Ontario Ministry of Training, Colleges, and Universities; the US ARO; and the Slovak Research and Development Agency grant APVV-0808-12 QIMABOS.


\bibliographystyle{myhamsplain}
\bibliography{momentumswitch}

\end{document}